\documentclass[conference]{IEEEtran}
\usepackage{spconf,amsmath,graphicx,hyperref,amssymb, amsthm}
\usepackage{xcolor}
\usepackage{subcaption}
\newcommand{\vc}[1]{{\mathbf #1}}
\newcommand{\vct}[1]{\pmb{#1}}
\newcommand{\ma}[1]{{\mathbf #1}}

\newtheorem{prop}{Proposition}

\newtheorem{definition}{Definition}
\usepackage[authormarkup=none]{changes}

\title{Graph Signal Surrogate Generation for Statistical Testing of Covariance Structure on Directed Graphs}

\name{Chun Hei Michael~Chan$^{1,2}$ \qquad Alexandre~Cionca$^{1,2}$ \qquad Dimitri~Van~De~Ville$^{1,2}$}
  
  \address{$^{1}$ Neuro-X Institute, Ecole Polytechnique Fédérale de Lausanne (EPFL)\\
      $^{2}$ Department of Radiology and Medical Informatics, University of Geneva}    

\begin{document}
%
\maketitle

\begingroup
\renewcommand\thefootnote{}\footnotetext{An extended version of this paper has been published as \cite{11626552}.}
\endgroup

\begin{abstract}
Non-parametric statistical testing is based on surrogate data generation that randomizes chosen features in the empirical data. In the graph setting, graph signal processing (GSP) brings forward versatile schemes; e.g., to preserve smoothness of graph signals as measured by the Dirichlet energy. However, how to deal with directed graphs remains an active area of research.
We begin by revisiting the definition of directed graph wide-sense stationarity. The surrogate signals preserve covariance under the stationary assumption. We demonstrate the feasibility of the scheme to detect irregular node covariance and benchmark our method against conventional schemes using the symmetrized graph. We also investigate how the level of asymmetry affects the detection performance, thus assessing the advantages of the presented approach. Finally, we show results for a real-world graph extracted from the Freeman EIES social network dataset.
\end{abstract}
\section{Introduction}
\label{sec:intro}
Graph Signal Processing (GSP) provides a powerful framework for analyzing signals on irregular domains, with broad applications in learning, neuroscience, and epidemiology~\cite{shuman_emerging_2013, marques_signal_2020}. Most existing GSP frameworks focus on undirected graphs. In particular, the notion of stationarity of graph signals has been extensively studied in the undirected setting~\cite{perraudin_stationary_2017, girault_translation_2015, marques_stationary_2017} and recently extended to directed graphs~\cite{iraji_wss_2025}. 

In conventional signal processing, a random process in the time domain is (weak) stationary when it is characterized by a constant first moment and time-invariant second moment. Equivalently, its covariance matrix is diagonalized by the Fourier basis 
\cite{perraudin_stationary_2017}. 
For undirected graphs, a similar definition is used, establishing graph wide sense stationary (GWSS) signals~\cite{perraudin_stationary_2017}. Specifically, a GWSS signal is stationary if and only if its covariance matrix is diagonalized by the eigenvectors of the graph shift operator (GSO). The extension of GWSS for directed graphs has been presented in recent work \cite{iraji_wss_2025}; i.e., the covariance matrix of a directed GWSS signal results from the left and right multiplication of a block diagonal matrix by respectively the generalized eigenvectors and its hermitian. This definition, however, considers the Jordan Normal Form (JNF) of the GSO, which raises two immediate concerns. First, the JNF is known to be computationally unstable and inefficient. Second, the JNF prevents phase operations \cite{chan_hilbert_2025}, which is central to the framework for surrogates presented here. Most graphs are however strongly connected leading to a diagonalizable GSO \cite{sevi2023harmonic}, and in the case where not, we will rely upon a recently proposed method \cite{seifert_digraph_2021} that iteratively dismantles Jordan blocks through rank-1 matrix perturbation to obtain diagonalizable GSO, therefore enabling the use of proper eigenvectors and eigenvalues.

The concept of stationarity on graphs has inspired several seminal contributions in GSP~\cite{iraji_wss_2025, cheung_graph_2018}. In particular, we proposed non-parametric statistical hypothesis testing for undirected graphs~\cite{pirondini_spectral_2016}, by exploiting sign randomization of the graph Fourier transform (GFT) coefficients to generate surrogate signals. A related surrogate generation scheme has also been introduced for graphs with complex Hermitian Laplacian representation~\cite{belda_new_2019}. This method is however not derived from stationary principles and is used in \cite{belda_new_2019} for data augmentation rather than statistical hypothesis testing. In both frameworks, the signal is randomized in the spectral domain, while preserving the variance of the initial graph Fourier coefficients. Finally, we emphasize that in contrast to the extensive literature on directionality inference in graph discovery~\cite{cui_topology_2024, shaska_causal_2025}, the present work is concerned with the statistical testing of signals observed on a fixed graph. A key distinction is that graphical causal models typically rely on acyclic graph assumptions, and the proposed framework applies to diagonalizable GSO, requiring cyclic graphs \cite{chan_hilbert_2025}.

Here, we first revisit the notion of directed GWSS through the eigendecomposition of the GSO. Next, we propose a graph phase randomization method to generate surrogate graph signals for non-parametric testing. We specifically assess the covariance, since second-order statistics are commonly used to infer interactions in networked systems \cite{bullmore2009complex}, yet may be confounded by the underlying graph structure. By testing covariance against structure-preserving surrogates, we isolate dependencies that cannot be explained solely by structural constraints. We compare our framework with undirected graph and permutation surrogates and evaluate detection performance on synthetic graphs with varying asymmetry levels and on a real social network. 

\section{Notations}
Lower bold cases denote vectors where $\vc x[n]$ indicates the $n$-th element of $\vc x$. Upper case bold letters are used for matrices where $\ma U[n,m]$ denotes the $(n,m)$-th entry of $\ma U$. The notations $(.)^*$, $(.)^T$, $(.)^H$, $(.)^{-1}$, $(.)^{-H}$ denote conjugate, transpose, hermitian, inverse, and hermitian inverse, respectively. We use $j$ to refer to the complex number such that $j^2=-1$. With $||\cdot||_2$, we refer to the $\ell_2$-norm.

\section{Preliminaries}
\label{sec:preliminaries}
\subsection{Undirected Graph}
Given an undirected graph with $N$ nodes represented by the symmetric adjacency matrix $\bar{\ma A}$ that we consider as the GSO, we also have the eigendecomposition
\begin{align*}
  \bar{\ma A}=\bar{\ma U}\bar{\ma \Lambda}\bar{\ma U}^T,  
\end{align*}
where $\bar{\ma U}$ are the real-valued eigenvectors with associated eigenvalues $\bar{\ma \Lambda}[n,n]$ on the diagonal of $\bar{\ma \Lambda}$, serves to define the GFT of a graph signal $\vc x$ which is defined as $\hat{\vc x}={\bar{\ma U}}^T {\vc x}$. In spectral domain the filtering \cite{shuman_emerging_2013} operation can be defined via the product of a diagonal matrix $\bar{\ma G}$  and the Fourier coefficients $\hat{\vc x}$ resulting in $\vc y =\bar{\ma U} \bar{\ma G}\hat{\vc x}=\bar{\ma U} \bar{\ma G}\bar{\ma U}^T\vc x$.

\subsection{Directed Graph}
For a directed graph with $N$ nodes, characterized by the (non-symmetric) adjacency matrix $\ma A$, the eigendecomposition yields:

\begin{align*}
  \ma A=\ma U\ma \Lambda \ma U^{-1},
\end{align*}
where the eigenvalues $\ma \Lambda[n,n]$ are either real-valued or complex conjugate pairs, with associated real-valued or complex conjugate eigenvectors, respectively. The GFT of a real-valued graph signal $\vc x$ is defined as $\hat{\vc x}=\ma U^{-1}{\vc x}$. For directed graphs, the filtering \cite{sandryhaila_discrete_2014} in spectral domain can similarly be defined through the product of a diagonal matrix $\ma G$  and the Fourier coefficients $\hat{\vc x}$ resulting in $\vc y =\ma U \ma G\hat{\vc x}=\ma U \ma G\ma U^{-1}\vc x$.

\section{Graph Stationarity}
\label{sec:stationarity}
In this section, we revisit the definition of GWSS as proposed in \cite{perraudin_stationary_2017, girault_translation_2015, iraji_wss_2025} through the eigendecomposition of the directed adjacency.

\subsection{Undirected Graph}
Let $\vc x \in \mathbb{R}^N$ be a random graph signal. Consider the first- and second-order statistical moments of $\vc x$ as the mean $\vct \mu_{\vc x}[n]=\mathbb{E}[\vc x[n]]$ and the covariance matrix $\ma H_{\vc x}=\mathbb{E}[({\vc x} - \vct \mu_{\vc x})({\vc x} - \vct \mu_{\vc x})^H]$. Without loss of generality, we assume throughout the paper that $\vct \mu_{\vc x}=\vc 0$. The definition of GWSS provided by \cite{perraudin_stationary_2017} yields the following equivalent characterization.

\begin{definition}[GWSS] \label{ustatio}
    A random graph signal $\vc x$ is GWSS if and only if its there exists a diagonal matrix $\hat{\bar{\ma H}}_{\vc x}\in\mathbb{R}^{N\times N}$ such that its covariance matrix can be expressed into:
\begin{align*} 
    \ma H_{\vc x}= \bar{\ma U}\hat{\bar{\ma H}}_{\vc x}\bar{\ma U}^T.
\end{align*}
In particular, the graph power spectral density (PSD) of $\vc x$ is defined to be $\hat{\bar{\vc h}}_{\vc x} \in\mathbb{R}^N$ with $\hat{\bar{\vc h}}_{\vc x}[n] = \hat{\bar{\ma H}}_{\vc x}[n, n]$ diagonal entries of $\hat{\bar{\ma H}}_{\vc x}$. 
\end{definition}

\subsection{Directed Graph}
We revisit GWSS processes on directed graphs, initially defined through the Jordan Decomposition~\cite{iraji_wss_2025}, using the eigendecomposition. In the case where an eigendecomposition of $\ma A$ exists, we instead have proper eigenvectors and a simpler formulation of the PSD. This brings us to the following definition of DGWSS.
\begin{definition}[DGWSS] \label{dstatio}
    A random graph signal $\vc x$ is Directed Graph Wide Sense Stationary (DGWSS) if and only if there exists a diagonal matrix $\hat{\ma H}_{\vc x}\in\mathbb{C}^{N\times N}$ such that its covariance matrix can be expressed as:
    \begin{align*}
    \ma H_{\vc x}= \ma U\hat{\ma H}_{\vc x}\ma U^H.
\end{align*}
In particular, the graph PSD of $\vc x$ is defined to be $\hat{\vc h}_{\vc x}\in\mathbb{C}^N$ with $\hat{\vc h}_{\vc x}[n] = \hat{\ma H}_{\vc x}[n, n]$ diagonal entries of $\hat{\ma H}_{\vc x}$.
\end{definition}

\begin{figure*}
    \centering
    \captionsetup[subfigure]{justification=centering}
    \subfloat[Synthetic graph\label{synthetic-graph}]{%
    \includegraphics[width=0.205\linewidth]{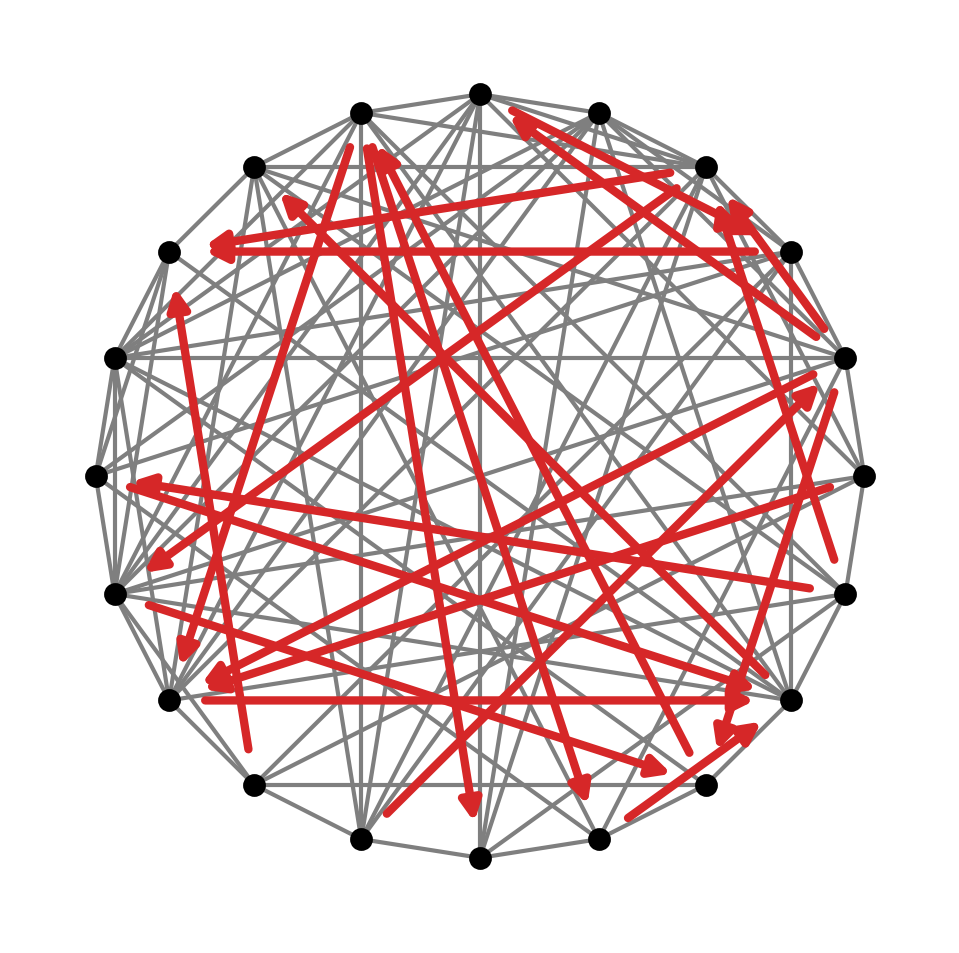}}
    \subfloat[True Irregularity \label{ground-truth}]{%
    \includegraphics[width=0.195\linewidth]{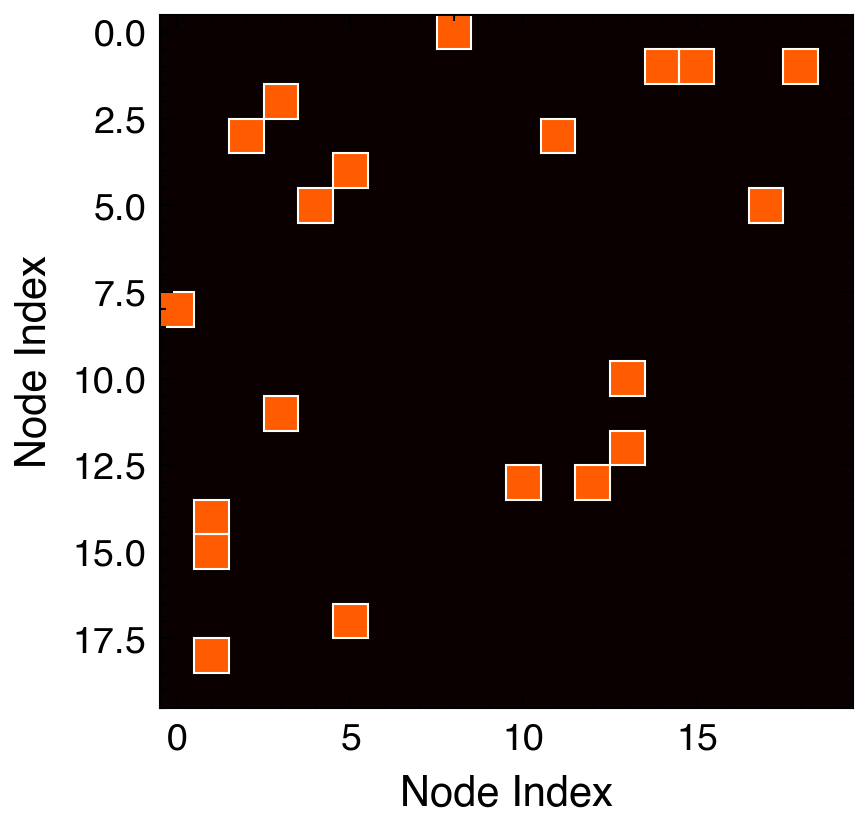}}
    \subfloat[ROC curves\label{roc-curve-synthetic}]{%
    \includegraphics[width=0.30\linewidth]{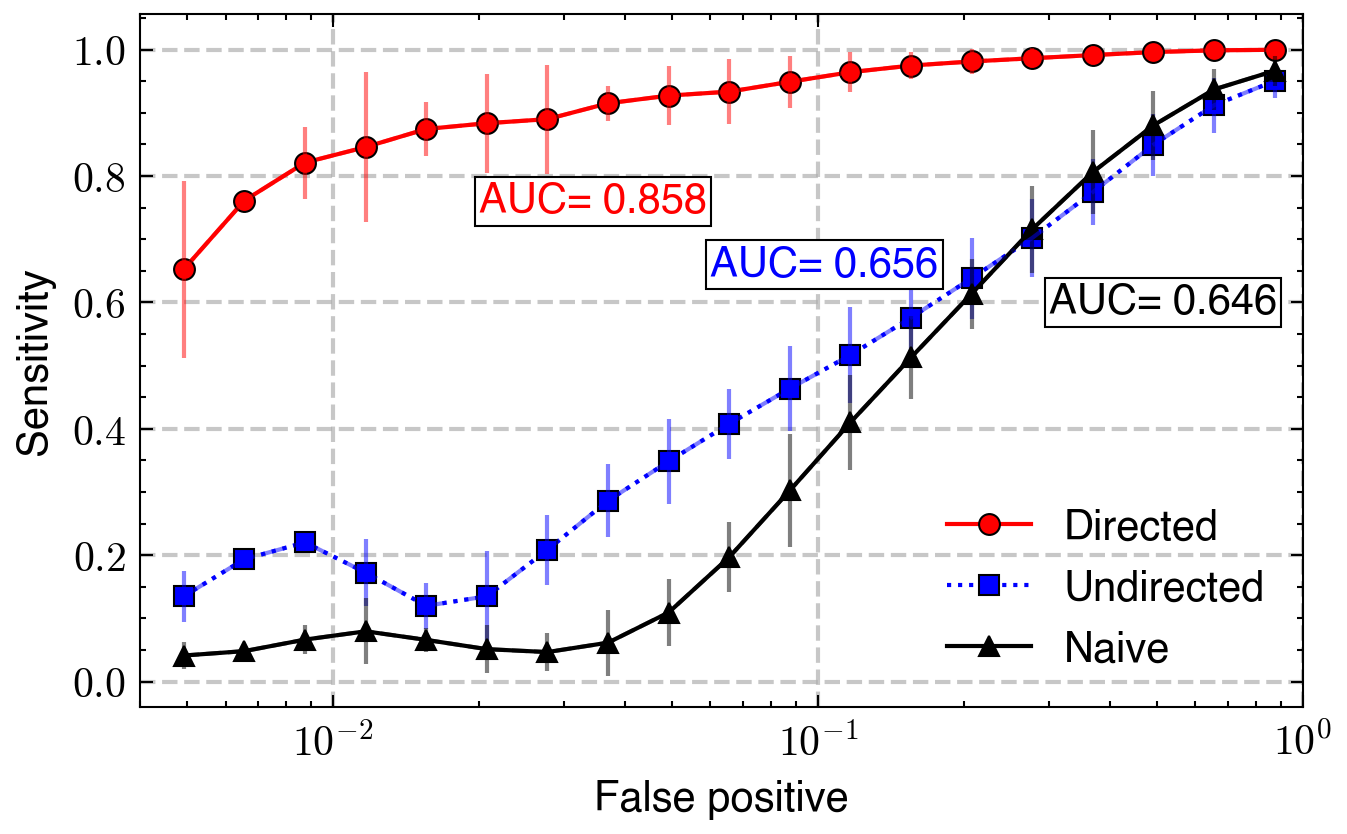}}
    \subfloat[AUC vs Asymmetry level\label{assymetrylevel}]{%
    \includegraphics[width=0.30\linewidth]{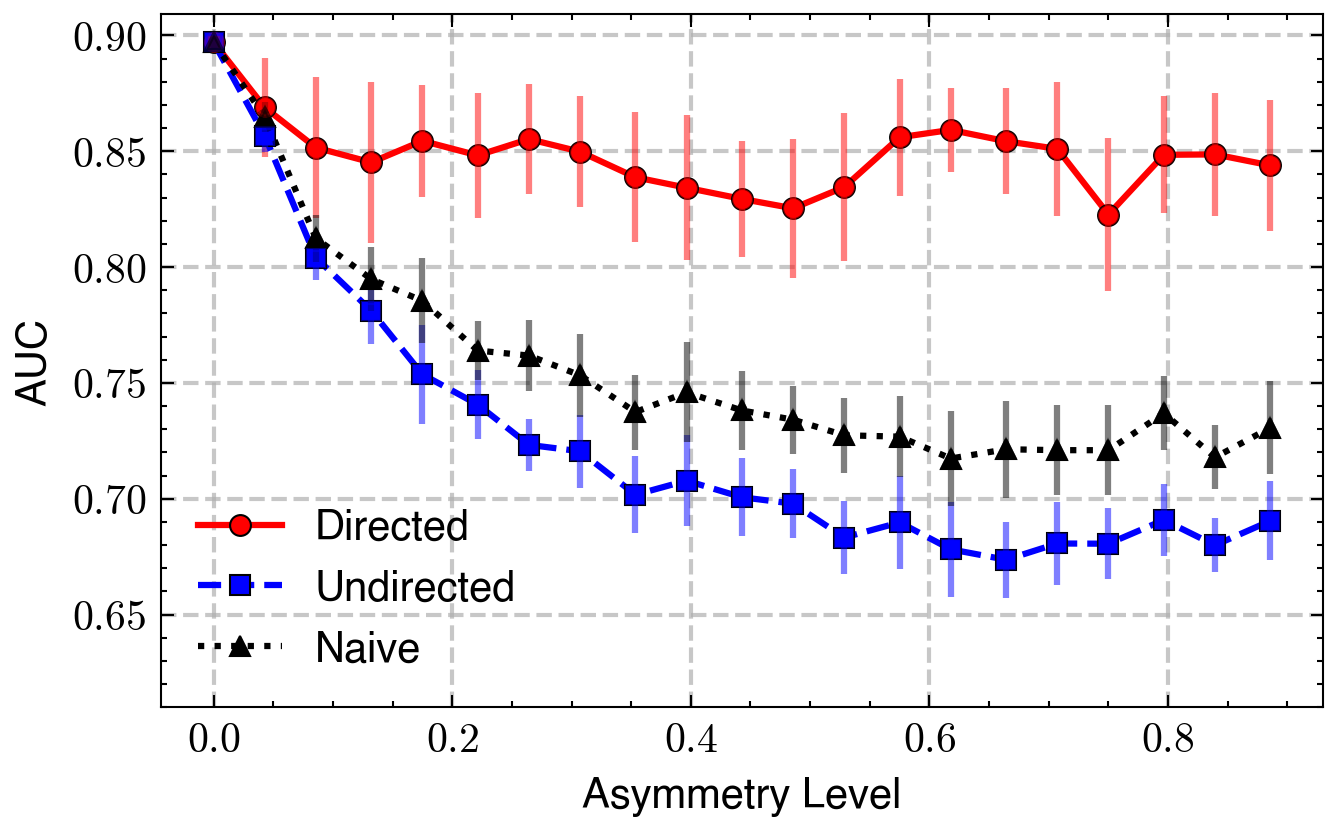}}
  \caption{a)~Example of synthetic graph with edge density $\rho=0.4$ and asymmetry level $\alpha=0.3$. Red edges are directed, while gray ones are undirected. b)~Example of ground truth to detect; i.e., non-zero entries of a random covariance matrix. c)~ROC curves averaged over $50$ covariance matrices considering graph in a). d) AUC scores averaged over $20$ graphs for different asymmetry levels.}
\end{figure*}

\section{Surrogates Generation Scheme}
\label{sec:surrogates}
We first enunciate the null hypothesis under which our surrogate graph signals are generated. Then, we briefly recall the randomization scheme on undirected graphs before generalizing the approach for directed graphs. Finally, we prove that the obtained surrogates are stationary signals.

\subsection{Null Hypothesis}
The statistical testing scheme used here is based on the acceptance/rejection of the null hypothesis $\mathcal{H}_0$: \emph{the empirical graph signals $\vc m_k$, $k=1,\ldots,K$, originate from a (D)GWSS process with respect to a given graph}.

In the following section, we will leverage the assumption under the null hypothesis to generate surrogate signals. 

\subsection{Surrogates on Undirected Graph}
Following our previous method~\cite{pirondini_spectral_2016}, we first define the randomization filter $\bar{\ma R}$ as the diagonal random-sign matrix
\begin{align*}
    \bar{\ma R}[n,n] = \epsilon_n,
\end{align*}
where $\epsilon_n$ is defined as $\epsilon_n\in\{-1,1\}$, and $Pr(\epsilon_n=1)=Pr(\epsilon_n=-1)=0.5$. Under the null, an empirical signal $\vc m$ is assumed to be a realization of a GWSS random signal on a given graph. Then, surrogates generated as \label{sign-randomize}
\begin{align*}
    \vc m^\text{surr}= \bar{\ma U} \bar{\ma R}\bar{\ma U}^T\vc m,
\end{align*}
preserve stationarity. 
\begin{prop}
    \label{prop-undirected-statio}
    Let $\vc m\in\mathbb{R}^N$ be a realization of a GWSS random graph signal w.r.t.\ a given graph, then randomizing the signs of the GFT coefficients preserves stationarity. 
\end{prop}
The proof of Proposition~\ref{prop-undirected-statio} is a particular case of the proof of Proposition~\ref{prop-directed-statio}.

\subsection{Surrogates on Directed Graphs}
On directed graphs, the GFT coefficients of graph signals become complex-valued, as well as the associated GFT frequencies (eigenvalues). Therefore, we leverage the presence of the phase to introduce a new randomization filter $\ma R$:
\begin{align*}
    \ma R[n,n] = 
    \left\{ \begin{array}{ll} 
    e^{j\Psi_n}, & \Im(\ma \Lambda[n,n]) \neq 0\\
    \epsilon_n, & \Im(\ma \Lambda[n,n])=0 
    \end{array} \right.
\end{align*}
with $\epsilon_n$ as in Sect.~\ref{sign-randomize} and 
$\Psi_n \overset{\text{i.i.d}}{\sim} \mathcal{U}([-\pi, \pi])$. If a real-valued graph signal is desired, then for GFT coefficients associated with conjugate pairs $n, m$, we set $\Psi_n=-\Psi_m$. 

In general, Parseval identity does not hold for the GFT of a directed graph, thus energy after phase randomization may not be preserved. However, since the observed signal $\mathbf{m}$ is DGWSS under the null, the signal energy remains the same after phase randomization as covariance is preserved and so is the trace of the covariance i.e the signal energy. Then, the surrogates are generated as
\begin{align*}
    \vc m^\text{surr} = \ma U \ma R \ma U^{-1}\vc m
\end{align*}
and preserve stationarity.
\begin{prop}
    \label{prop-directed-statio}
    Let $\vc m\in\mathbb{R}^N $ be a realization of a DGWSS random graph signal w.r.t.\ a given graph, then phase-randomizing the GFT coefficients preserves stationarity.
\end{prop}
\begin{proof}
We set out to show that phase-randomizing $\vc m$ yields again a realization of a DGWSS graph signal. Let $\vc x$ and $\vc s$ be the random variables associated to the realization $\vc m$ and its phase-randomized counterpart $\vc m^\text{surr}$, respectively. First note that $\mathbb{E}[\vc s]=\ma U\mathbb{E}[\ma R]\ma U^{-1}\mathbb{E}[\vc x]=\vc 0$ as $\ma R$ is defined independently from $\vc x$.
Then, expanding the covariance of $\vc s$ we have
\begin{align*}
\ma H_{\vc s} = \mathbb{E}[\ma U \ma R \underbrace{\ma U^{-1} {\vc x\vc x}^H\ma U^{-H}}_{\hat{\ma H}_\vc x} \ma R^H\ma U^H].
\end{align*}
With $\hat{\ma H}=\mathbb{E}[\ma R\hat{\ma H}_\vc x\ma R^H]$ and its entries
\begin{align*}
   \hat{\ma H}[m,n] = \mathbb{E}\big[\ma R[m,m]{\bf R}^*[n,n]\big] \hat{\ma H}_\vc x[m,n],
\end{align*}
one can observe in particular that $\mathbb{E}\big[\ma R[m,m]\ma R^*[n,n]\big]=\delta_{mn}$ where $\delta_{mn}$ denotes the Kronecker delta. 
Finally we have
\begin{align*}
    \ma H_\vc s=\ma U \hat{\ma H} \ma U^H,
\end{align*}
with $\hat{\ma H}=\text{diag}(\hat{\ma H}_\vc x)$ diagonal matrix. Consequently, $\vc s$ is DGWSS, and phase-randomizing $\vc m$ indeed yields a realization of a DGWSS graph signal, which concludes the proof.
\end{proof}

\subsection{Testing with Surrogates}
We want to assess whether the observed covariance satisfies the null hypothesis. From the empirical data $\{\mathbf{m}_k\}_{k=1}^K$, we can build the covariance matrix as
\begin{align*}
\ma H = \frac{1}{K} \sum_{k=1}^K \vc m_k \, \vc m_k^T.    
\end{align*}
To test the elements $\ma H[n,m]$, we generate $T$ surrogate sets 
$\{{\vc m}_{k}^{\text{surr},(t)},k=1,\ldots,K\}_{t=1}^T$. For each surrogate set, the covariance is computed, from which the null distribution of each element can be derived. In particular, all elements are ranked as $\ma H^{\text{surr},(t)}$, $t=1,\ldots,T$, and compared against the empirical covariance $\ma H[n,m]$: the $p$-value under the null corresponds to $t/(T+1)$ when its value falls in the $t$-th bin. For detection, multiple comparison correction can be performed using the Bonferroni method or recent multiple hypothesis testing approach leveraging graph structure \cite{jian2025graph}.

\section{Experimental Results}
\label{sec:experiments}
To validate the proposed scheme, we consider empirical signal as follows: 
\begin{align*}
    \vc m_k[n] = \vct \epsilon_k[n] + \vct \eta_k[n]\ , \ \ k=1,...,K,
\end{align*}
with $\vct \epsilon_k$ realizations of a DGWSS on a given graph, and of mean $\vc 0$; $\vct \eta$ a random signal of mean $\vc 0$ and covariance $\ma H_{\vct \eta}$. The observations $\vc m_k$ are sampled from a multivariate gaussian. These realizations violate the null due to the contribution of $\vct \eta_k$. The goal is to detect those deviations. Practically, the matrix $\ma H_{\vct \eta}$ causes the irregular covariance of $\ma H$ and we focus on the cross covariance of $\ma H$, thus in subsequent experiments non-zero off-diagonal entries of $\ma H_{\vct \eta}$ will act as ground truth for detection. We take $\ma H_{\vct \epsilon}$ as the covariance of a directed graph white noise, resulting in graph-dependant covariances~\cite{iraji_wss_2025}. Throughout, entries of the covariance $\ma H_{\vct \eta}$ are randomly set to $0$ or $1$ and subsequently rescaled to achieve a variance of $0.5$ corresponding to approximatively half the variance of $\vct \epsilon$ across different graphs, all while ensuring positive definiteness. All code and data used in this article are openly available\footnote{https://github.com/MIPLabCH/Graph-Test-Covariance}.

\begin{table}
\centering
\caption{AUC Scores for covariance detection on synthetic graph with density $\rho=0.5$ and asymmetry level $\alpha=0.35$.}
\begin{tabular}{lcc}
\hline
\textbf{Model} & \textbf{Synthetic} & \textbf{Real} \\ & ($\rho=0.5$, $\alpha=0.35$) \\
\hline
Directed      & \textbf{0.84 $\pm$ 0.03} & \textbf{0.82} \\
Undirected    & 0.77 $\pm$ 0.02 & 0.75 \\
Permutation   & 0.76 $\pm$ 0.02 & 0.81 \\
\hline
\end{tabular}
\label{tab:auc_scores}
\end{table}

\subsection{Synthetic Graph: Generation}
\label{sec:graph-generation}
To generate synthetic graphs, we proceed in a similar fashion to \cite{flaxman2007expansion}, which is to start from a basic graph: in our case the undirected cyclic graph. We then introduce directed and undirected edges between random uniformly selected pairs of nodes until a given density of edges $\rho$ is reached. To vary the asymmetry of the graph, we control the proportion of directed edges through the probability of the added edge being directed, otherwise the added edge is undirected. To evaluate the asymmetry level $\alpha$ of a graph we compute the ratio of the number of directed edges to the total number of edges in the graph. We see an example of such a generated graph in Fig.~\ref{synthetic-graph} for density $\rho=0.4$ and asymmetry level $\alpha=0.3$.

\begin{figure}
    \centering
    \captionsetup[subfigure]{justification=centering}
    \subfloat[Real graph\label{real-graph}]{%
    \includegraphics[width=0.41\linewidth]{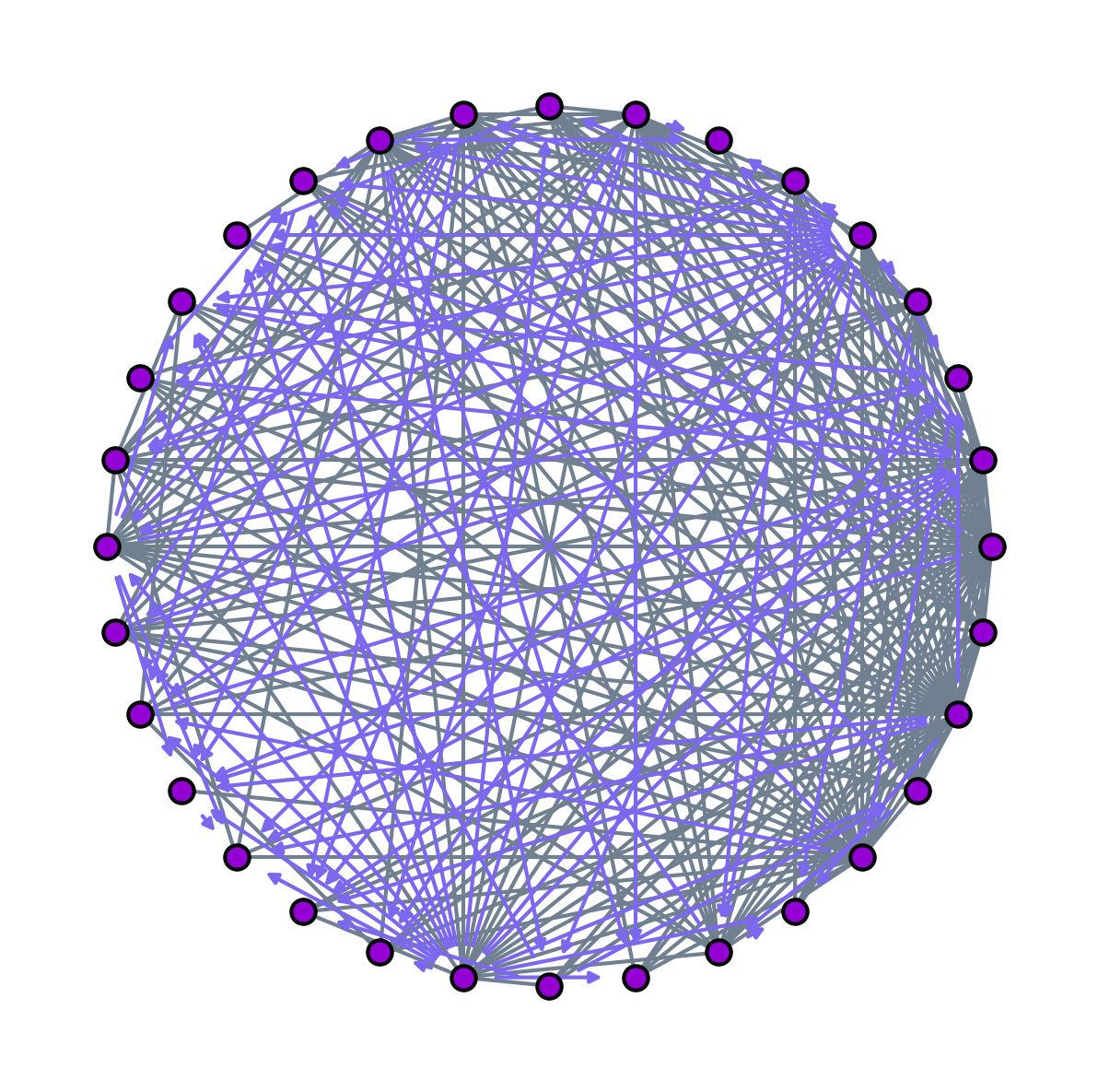}}
    \subfloat[ROC curves\label{real-roc}]{%
    \includegraphics[width=0.59\linewidth]{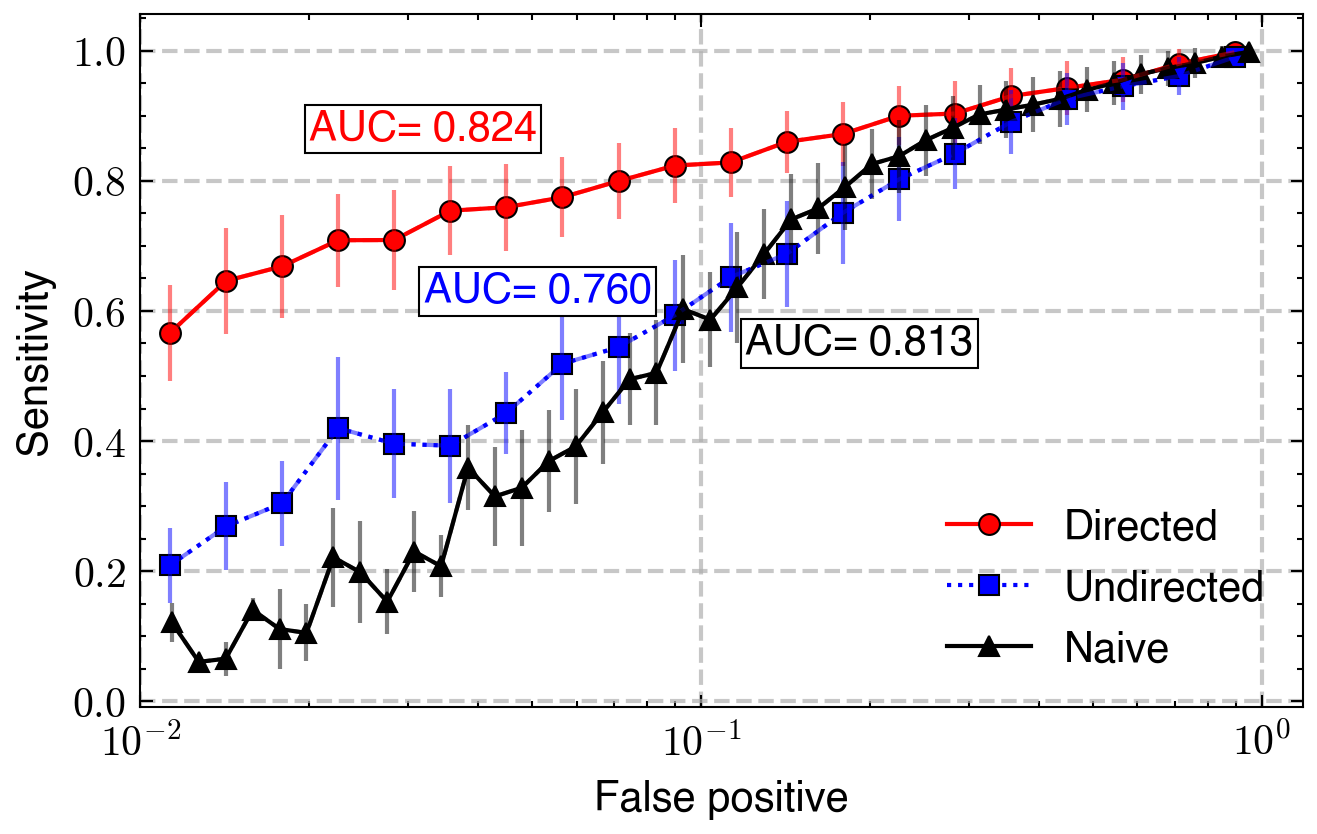}}
  \caption{a) Real graph extracted from Freeman EIES dataset. Purple edges are directed while gray ones are undirected. b) ROC curves over $50$ covariance matrices.}
\end{figure}

\subsection{Synthetic Graph: Baseline Performance} 
\label{base-case}
We employ a synthetic graph of $N=20$ nodes, generated as detailed in Sect.~\ref{sec:graph-generation}. We also show an example of ground truth irregularity to be detected in Fig.~\ref{ground-truth}. To perform testing, we generate $999$ surrogates and compute detection sensitivity and false positives at various significance thresholds to obtain a ROC curve. Repeating this procedure for $50$ random covariances of $\vct \eta$, we obtain multiple ROC curves and an AUC score on the mean ROC curve. The results are presented in Fig.~\ref{roc-curve-synthetic} and are compared to permutation-based surrogates, obtained by randomly shuffling the entries of the observed graph signals, as well as undirected graphs surrogates, where the graph is defined using the symmetrized adjacency $\frac{A+A^T}{2}$.

\subsection{Synthetic Graph: Varying Asymmetry level}
\label{sec:varying-asymmetry}
We take $\rho=0.5$ (chosen to be consistent with the real graph in Sect.~\ref{sec:real-graph}) and vary the graphs' asymmetry level $\alpha$ from $0$ to $0.9$, as the base graph is undirected, for which we each generate $50$ graphs. For each of these graphs, we proceed with the same setting as in Sect.~\ref{base-case} to obtain ROC curves (Fig.~\ref{roc-curve-synthetic}) and summarizing AUC scores. Finally we can observe in Fig.~\ref{assymetrylevel} the impact of varying asymmetry levels on the detection performance. Detection performance remains constant in directed models, while in undirected and permutation-based models, the performance drops as asymmetry level increases.

\subsection{Real-World Graph}
\label{sec:real-graph}
We use the Freeman EIES dataset \cite{freeman1979networkers}, recording the number of messages sent among $N=32$ researchers. To construct the graph, we set the adjacency matrix entry $(m,n)$ to be whether the $m$-th researcher (node) often sends messages to the $n$-th researcher, that is, when the number of messages sent between the pair is greater than the median of all numbers of messages. As the number of messages from $m$ to $n$ may differ from those from $n$ to $m$, the binary matrix representation is asymmetric. Fig.~\ref{real-graph} shows the graph structure displaying an edge density of $\rho=0.5$ and an asymmetry level of $\alpha=0.35$. We then proceed with the same testings as in Sect.~\ref{base-case} to obtain ROC curves (Fig.~\ref{real-roc}) and an AUC score. We observe that directed graph surrogates again outperform undirected and permutation-based ones. Moreover, the AUC scores obtained are consistent with synthetic example of similar edge density and asymmetry level (Table~\ref{tab:auc_scores}).




\section{Conclusion}
\label{sec:conclusion}
We presented a novel approach for non-parametric statistical testing of signals on directed graphs. Our experiments demonstrated the feasibility of our approach, highlighting the importance of accounting for graph directionality, and provided a simple way to gauge detection performance in real graphs of arbitrary asymmetry. Future work will focus on extending surrogate-based testing beyond irregular covariance detection and exploring practical applications.

\section{Acknowledgement}
This work was supported by the Swiss National Science Foundation Sinergia 209470 ``Precision mapping of electrical brain network dynamics with application to epilepsy''. 


\bibliographystyle{ieeetr}
\bibliography{strings,refs}

\end{document}